\documentclass[conference]{IEEEtran}
\IEEEoverridecommandlockouts
\usepackage{multicol}
\usepackage{cite}
\usepackage{amsmath}
\usepackage{amsfonts}
\usepackage{pifont}
\usepackage{amssymb}
\usepackage{amsthm}
\usepackage{tikz}
\usepackage{cases}
\usepackage{graphicx}
\usepackage{float,stfloats}
    \graphicspath{{../}}
    \DeclareGraphicsExtensions{.pdf}
\usepackage[justification=centering]{caption}
\usepackage[caption=false,font=footnotesize]{subfig}
\usepackage{multirow}
\usepackage{booktabs}
\usepackage{url}
\usepackage{xtab}
\usepackage{tabu}
\usepackage{longtable}
\usepackage{algorithm}
\usepackage{algorithmic}
\usepackage{enumerate}
\usepackage{makecell}
\usepackage{lipsum}
\usepackage{multicol}
\usepackage{mathdots}
\usepackage{extarrows}
\usepackage{color,xcolor}
\usepackage{bm}
\usetikzlibrary{arrows}
\usepackage{caption}
\usepackage[numbers,sort&compress]{natbib}
\theoremstyle{plain}

\newtheorem{theorem}{Theorem}
\newtheorem{lemma}[theorem]{Lemma}

\newtheorem{corollary}[theorem]{Corollary}

\theoremstyle{definition}
\newtheorem{definition}{Definition}

\newtheorem{construction}{Construction}

\newtheorem{remark}{Remark}

\def\BibTeX{{\rm B\kern-.05em{\sc i\kern-.025em b}\kern-.08em
		T\kern-.1667em\lower.7ex\hbox{E}\kern-.125emX}}

\begin{document}
	\title{A Formula for the I/O Cost of Linear Repair Schemes and Application to Reed-Solomon Codes}
	
	\author{\IEEEauthorblockN{Zhongyan Liu, Zhifang Zhang}\\
\IEEEauthorblockA{\fontsize{9.8}{12}\selectfont KLMM, Academy of Mathematics and Systems Science, Chinese Academy of Sciences, Beijing 100190, China\\
School of Mathematical Sciences, University of Chinese Academy of Sciences, Beijing 100049, China\\
Emails: liuzhongyan@amss.ac.cn, zfz@amss.ac.cn}
}
	\maketitle
	
	\thispagestyle{empty}

\begin{abstract}
Node repair is a crucial problem in erasure-code-based distributed storage systems. An important metric for repair efficiency is the I/O cost which equals the total amount of data accessed at helper nodes to repair a failed node. In this work, a general formula for computing the I/O cost of linear repair schemes is derived from a new perspective, i.e., by investigating the Hamming weight of a related linear space. Applying the formula to Reed-Solomon (RS) codes, we obtain lower bounds on the I/O cost for full-length RS codes with two and three parities. Furthermore, we build linear repair schemes for the RS codes with improved I/O cost. For full-length RS codes with two parities, our scheme meets the lower bound on the I/O cost.
 \end{abstract}
\begin{IEEEkeywords}
	Reed-Solomon codes, optimal access, distributed storage system.
\end{IEEEkeywords}

\section{Introduction}\label{Sec1}
Erasure codes, especially the Maximum Distance Separable (MDS) codes, are extensively used in large-scale distributed storage systems (DSS) to ensure fault-tolerant storage with low redundancy. Specifically, the original file consisting of $k$ blocks is encoded into $n$ blocks using an MDS code, and then stored across $n$ nodes such that any $k$ nodes suffice to recover the original file. A crucial problem in the erasure-code-based DSS is the efficient node repair from helper nodes. Suppose each block is composed of $\ell$ data units. During the repair process, each helper node transmits $\leq \ell$ units computed from its storage. Two important metrics for the repair efficiency are the total amount of data transmitted during the repair process (i.e., repair bandwidth) and the volume of data accessed at the helper nodes (i.e., I/O cost).

In \cite{csbound}, Dimakis et al. proved the cut-set bound to characterize the optimal repair bandwidth. Then a lot of work are devoted to design linear array codes that meet the cut-set bound \cite{RCviaPM,Zigzag,Ye-OAMSR,Li-OAMSR}, some of which even attain the optimal I/O cost at the same time \cite{Ye-OAMSR,Li-OAMSR}. However, these specially designed array codes are less appealing in industry because the infrastructure has been adapted to some classical MDS codes, such as Reed-Solomon (RS) codes.
Therefore, it is more significant to develop efficient repair schemes for currently used erasure codes.

Shanmugam et al. \cite{scalarMDS} first  established a framework of linear repair schemes for scalar MDS codes. Then Guruswami and Wootters \cite{RSrepair} derived a characterization of linear repair schemes from dual codewords. They also designed a repair scheme for $[n,k]$ RS codes over $\mathbb{F}_{q^\ell}$ with $n-k\geq q^{\ell-1}$, and proved their scheme achieves the optimal repair bandwidth for the full-length RS codes. This work was later extended to $n-k\geq q^s$ for any $s<\ell$ in \cite{obRS}. Although the repair bandwidth in \cite{RSrepair,obRS} is optimal for the full-length case, it is far from the cut-set bound due to the limitation $\ell={\rm O}(\log n)$. For sufficiently large $\ell\gtrsim n^n$, Tamo et al. \cite{RScsbound1} designed a family of RS codes with the repair bandwidth  achieving the cut-set bound.

How well RS codes perform when taking into account the I/O cost is an open problem raised by Guruswami and Wootters \cite{RSrepair}. Computing the I/O cost of existing repair schemes is already a difficult problem, let alone designing repair schemes with lower I/O cost. In \cite{IOfulllength}, the authors showed the repair schemes in \cite{RSrepair} and \cite{obRS} for full-length RS codes incur a trivial I/O cost, i.e., the entire file needs to be read to repair a single node. Dau et al. considered the case of full-length RS codes with two parities in \cite{fullr=2}. They derived a lower bound on the I/O cost and provided repair schemes that achieve this bound. However, their bound and repair schemes only work for the field of characteristic $2$. Then, Li et al. \cite{shortr=2} extended the lower bound  to the short-length RS codes (i.e., the evaluation-set is a subspace of $\mathbb{F}_{q^\ell}$) with two parities. Also, their lower bound only works for $q=2$, and no matching repair schemes are given. In \cite{oaRS}, Chen et al. enhanced the repair scheme in \cite{RScsbound1} with the optimal access property (i.e., the I/O cost equals the repair bandwidth, both matching the cut-set bound), by further enlarging $\ell$ by a factor exponential in $n$. Unfortunately, such a large $\ell$ is infeasible in practical use.

\subsection{Contributions}
In this work, we study the I/O cost of linear repair schemes for RS codes. Our contributions are as follows.
\begin{enumerate}
    \item {\bf A general formula for the I/O cost of linear repair schemes.}
    Before our work, the authors in \cite{IOfulllength} and \cite{fullr=2} computed the I/O cost directly from the definition where the volume of data accessed at each helper node is calculated independently. Our formula relates the I/O cost to the Hamming weight of a linear space involving all nodes which makes the summation easier to compute.
    \item {\bf Lower bounds on the I/O cost for RS codes with two and three parities.} Compared to the bound derived in \cite{fullr=2}, we remove the restriction of field characteristic $2$ and provide a much simpler proof. Along the line, we also obtain a lower bound on the I/O cost for the full-length RS codes with three parities.
    \item {\bf Repair schemes for RS codes with lower I/O cost.} Our constructions towards reducing the I/O cost are inspired from the formula. For the case of full-length and two parities, our scheme matches the proved lower bound and thus reaches the optimal I/O cost. For three parities and $q=2$, our scheme accesses $2^{\ell-3}$ bits higher than the proved lower bound. However, the I/O cost is still $2^\ell-\ell$ bits lower than that of the scheme in \cite{obRS}. In general, our scheme has the repair bandwidth equal to the I/O cost, which implies a repair-by-transfer process. Moreover, the sum of repair bandwidth and I/O cost of our scheme is lower than the schemes proposed in \cite{RSrepair} and \cite{obRS}, which means the reduction of I/O cost is more significant than the sacrifice in the repair bandwidth.
    \end{enumerate}
The remaining of the paper is organized as follows. Section \ref{Sec2} introduces preliminaries of repair schemes. Section \ref{Sec3} presents the formula for the I/O cost. Then Section \ref{Sec4} proves the lower bounds and Section \ref{Sec5} presents the repair scheme. Finally Section \ref{Sec6} concludes the paper.

\section{Preliminaries}\label{Sec2}
For positive integers $m\leq n$, denote $[n]=\{1,...,n\}$ and $[m,n]=\{m,m+1,...,n\}$. Let $B=\mathbb{F}_q$ be the finite field of $q$ elements and $F=\mathbb{F}_{q^\ell}$ be the extension field of $B$ with degree $\ell$. It is known that $F$ is an $\ell$-dimensional linear space over $B$. Thus, elements in $F$ are also viewed as vectors of length $\ell$ over $B$. More specifically, let $\mathcal{B}=\{\beta^{(1)},...,\beta^{(\ell)}\}$ be a basis of $F$ over $B$ and $\hat{\mathcal{B}}=\{\gamma^{(1)},...,\gamma^{(\ell)}\}$ be the dual basis of $\mathcal{B}$, i.e., ${\rm Tr}(\beta^{(i)}\gamma^{(j)})={\bm 1}_{i=j}$, where ${\rm Tr}(\cdot)$ is the trace function of $F$ over $B$, then it holds $\alpha=\sum_{i=1}^\ell{\rm Tr}(\alpha\gamma^{(i)})\beta^{(i)}$ for any $\alpha\in F$. Define a map $\Phi_{\mathcal{B}}$ from $F$ to $B^\ell$ such that \begin{equation}\Phi_{\mathcal{B}}(\alpha)=({\rm Tr}(\alpha\gamma^{(1)}),...,{\rm Tr}(\alpha\gamma^{(\ell)}))\;,~\forall \alpha\in F.\label{eq1-}\end{equation}
Actually, $\Phi_{\mathcal{B}}$ is a $B$-linear bijection and thus called the vector representation of elements of $F$ over $B$ with respect to the basis $\mathcal{B}$. For simplicity, we also define $\Phi_{\mathcal{B}}({\bm \alpha})=(\Phi_{\mathcal{B}}(\alpha_1),...,\Phi_{\mathcal{B}}(\alpha_n))$ for any ${\bm \alpha}=(\alpha_1,...,\alpha_n)\in F^n$.

 All vectors throughout are treated as row vectors and denoted by bold letters. For a vector of length $m$, say, $\bm{x}=(x_1,...,x_m)$, define ${\rm supp}(\bm{x})=\{j\in[m]:x_j\neq 0\}$ and ${\bf wt}(\bm{x})=|{\rm supp}(\bm{x})|$. Furthermore, for a set of vectors $W\subseteq B^m$,  define $\mathrm{supp}(W)=\bigcup_{\bm{x}\in W}\mathrm{supp}(\bm{x})$ and ${\bf wt}(W)=\sum_{\bm{x}\in W}{\bf wt}(\bm{x})$. Additionally, let ${\rm span}_{B}(W)$ be the $B$-linear space spanned by the vectors in $W$, and ${\rm dim}_{B}(W)$ denote the corresponding rank over $B$.

Throughout we study the repair schemes of Reed-Solomon codes defined over \!$F$ and the I/O cost (also the repair bandwidth) is measured in the number of elements in $B$. For simplicity, the elements in $F$ are called symbols and those in $B$ are called subsymbols.

\subsection{Linear repair scheme for MDS codes}
%In \cite{scalarMDS}, Shanmugam et al. developed a framework of repair schemes for MDS codes by vectorizing the symbols stored in each node and transmitting only subsymbols during the repair process.
In \cite{RSrepair}, Guruswami and Wootters derived a characterization of the linear repair schemes. Specifically, they proved the following theorem.
\begin{lemma}[Guruswami-Wootters \cite{RSrepair}]\label{lem1}
Let $\mathcal{C}$ be an $[n,k]$ MDS code over $F$, and $B\subseteq F$ be a subfield of $F$ with $[F:B]=\ell$. Then for each node $i^*\in[n]$, the following are equivalent.
\begin{enumerate}
  \item[(1)] There is a linear repair scheme for node $i^*$ over $B$ with bandwidth $b$.
  \item[(2)] There exist $\ell$ dual codewords ${\bm g}^{(1)},...,{\bm g}^{(\ell)}\in\mathcal{C}^\bot$, where ${\bm g}^{(j)}\!\!=\!\!({\bm g}_{1}^{(j)},...,{\bm g}_{n}^{(j)})\!\in\! F^n$\! for\! $j\!\!\in\!\![\ell]$, such that ${\rm dim}_B\big(\{{\bm g}_{i^*}^{(j)}\}_{j=1}^\ell\big)=\ell$, and $$b=\sum_{i\in[n]\setminus\{i^*\}}{\rm dim}_B\big(\{{\bm g}_{i}^{(j)}\}_{j=1}^\ell\big).$$
\end{enumerate}
\end{lemma}
In more detail, for any codeword ${\bm c}=({\bm c}_1,...,{\bm c}_n)\in\mathcal{C}$, it holds ${\bm c}_{i^*}{\bm g}^{(j)}_{i^*}=-\sum_{i\neq i^*}{\bm c}_{i}{\bm g}^{(j)}_{i}$ for $j\in[\ell]$. Applying the trace function to both sides of the identities, one obtains
\begin{equation}\label{eq1}
{\rm Tr}({\bm c}_{i^*}{\bm g}^{(j)}_{i^*})=-\sum_{i\neq i^*}{\rm Tr}({\bm c}_{i}{\bm g}^{(j)}_{i}),~~\forall j\in[\ell]\;.
\end{equation}
If ${\rm dim}_B(\{{\bm g}_{i^*}^{(j)}\}_{j=1}^\ell)=\ell$, then the $\ell$ traces on the left side of the identities in (\ref{eq1}) suffice to recover ${\bm c}_{i^*}$. To this end, each helper node $i\in[n]\setminus\{i^*\}$ transmits $\{{\rm Tr}({\bm c}_{i}{\bm g}^{(j)}_{i})\}_{j=1}^\ell\subseteq B$.

Since the dual code is a more common object in coding theory, linear repair schemes, especially those for scalar MDS codes such as RS codes, are mostly described by the dual codewords in the literature. In this work, we focus on the I/O cost of these repair schemes. Next we review the definition of the I/O cost given in \cite{IOfulllength} and restate some basic results about the I/O cost using the notations defined above.

\begin{definition}[\cite{IOfulllength}]
The (read) I/O cost of a function $f(\cdot)$ with respect to a basis $\mathcal{B}$ is the minimum number of subsymbols of $\alpha\in F$ needed to compute $f(\alpha)$. The I/O cost of a set of functions $\mathcal{F}$ is the minimum number of subsymbols of $\alpha$ needed for the computation of $\left\{f(\alpha):f\in\mathcal{F}\right\}$.
\end{definition}

\begin{lemma}[\cite{IOfulllength}]The following statements hold.
\begin{enumerate}
    \item The I/O cost of the trace functional  $\mathrm{Tr}^{\gamma}(\alpha)\triangleq\mathrm{Tr}(\gamma\alpha)$ with respect to $\mathcal{B}$ is ${\bf wt} \big(\Phi_{\hat{\mathcal{B}}}(\gamma)\big)$.
    \item The I/O cost of the set of trace functionals $\{\mathrm{Tr}^{\gamma}(\cdot):\gamma\in\Gamma\}$ with respect to $\mathcal{B}$ is $|\mathrm{supp}(\{\Phi_{\hat{\mathcal{B}}}(\gamma):\gamma\in\Gamma\})|.$
\end{enumerate}
\end{lemma}

\begin{lemma}[\cite{shortr=2}]\label{lem3}
The I/O cost of the repair scheme of node $i^*$ based on a set of dual codewords $\{{\bm g}^{(j)}\}_{j=1}^\ell\subseteq \mathcal{C}^\bot$ with respect to $\mathcal{B}$ is $\gamma_{I/O}=\sum_{i\in[n]\setminus\left\{i^{*}\right\}}\mathrm{nz}(W_i)$, where $\mathrm{nz}(W_i)$ specifies the number of nonzero columns in the $\ell\times\ell $ I/O matrix $W_i$ defined as
\fontsize{8pt}{2pt}
 { \begin{align}\label{W_j}
   W_i\!\!=\!\!\!\left(\!{
             \setlength{\arraycolsep}{1pt}\begin{array}{c}
               \Phi_{\hat{\mathcal{B}}}({\bm g}_{i}^{(1)})  \\
               \Phi_{\hat{\mathcal{B}}}({\bm g}_{i}^{(2)}) \\
                  \vdots\\
               \Phi_{\hat{\mathcal{B}}}({\bm g}_{i}^{(\ell)})
             \end{array}
          }\!\right)
          \!\!=\!\!\left(\! {\setlength{\arraycolsep}{1pt}\begin{array}{cccc}\mathrm{Tr}({\bm g}_{i}^{(1)}\beta^{(1)})&\mathrm{Tr}({\bm g}_{i}^{(1)}\beta^{(2)}) &  \cdots &\mathrm{Tr}({\bm g}_{i}^{(1)}\beta^{(\ell)})\\
\mathrm{Tr}({\bm g}_{i}^{(2)}\beta^{(1)})& \mathrm{Tr}({\bm g}_{i}^{(2)}\beta^{(2)})& \cdots &\mathrm{Tr}({\bm g}_{i}^{(2)}\beta^{(\ell)})\\
 \vdots & \vdots & \ddots & \vdots \\
 \mathrm{Tr}({\bm g}_{i}^{(\ell)}\beta^{(1)})&\mathrm{Tr}({\bm g}_{i}^{(\ell)}\beta^{(2)})&\cdots &\mathrm{Tr}({\bm g}_{i}^{(\ell)}\beta^{(\ell)})\end {array}}\!\right).
\end{align}}
\end{lemma}

It can be seen  ${\rm rank}(W_i)\!=\!{\rm dim}_B\big(\{{\bm g}_{i}^{(j)}\}_{j=1}^\ell\big)$, so the repair bandwidth can be calculated as $b=\sum_{i\in[n]\setminus\left\{i^{*}\right\}}\mathrm{rank}(W_i)$. Moreover,
the I/O cost at each helper node is computed with respect to a fixed basis (which may be different from the basis used by other nodes). We assume throughout this work that all nodes use a common basis $\mathcal{B}$.

\section{The Formula of I/O Cost}\label{Sec3}
In this section, we derive a general formula (Theorem \ref{I/O}) for the I/O cost of the linear repair scheme based on dual codewords. Instead of counting the number of nonzero columns of the I/O matrices as in \cite{fullr=2} and \cite{shortr=2}, we reduce the problem to the calculation of Hamming weight of a linear space which is easier to handle in most cases.

We begin with a lemma that explores the relationship between the weight and support of a linear space, which is frequently used in subsequent proofs and constructions.

\begin{lemma}\label{wt}
Let $G$ be a $k\times m$ matrix over $B$ whose rows are linearly independent. Define $W\!=\!\{{\bm u}G:{\bm u}\!\in\! B^k\}$ and suppose $\bm{y}=(y_1,...,y_m)\in B^m$.

\noindent(i) $|{\rm supp}(W)|={\rm nz}(G)$.

\noindent(ii) If $j\in \mathrm{supp}(W)$, then each element of $B$ will appear $q^{k-1}$ times in the $j$-th position of $\bm{y}+W$; if $j\not\in \mathrm{supp}(W)$, then the $j$-th position of $\bm{y}+W$ is $y_j$.

\noindent(iii)
\resizebox{0.93\columnwidth}{!}{${\bf wt}({\bm y}\!+\!W)\!=\!|{\rm supp}(W)|q^{k\!-\!1}\!(q\!-\!1)\!+\!|{\rm supp}({\bm y})\!\setminus\!{\rm supp}(W)|q^k.$} Particularly, ${\bf wt}(W)\!=\!|{\rm supp}(W)|q^{k-1}(q\!-\!1)\!\geq\! kq^{k-1}(q\!-\!1)$.

\end{lemma}

\begin{proof}
Denote $G=\left(\bm{w}_1^\top,\bm{w}_2^\top,...,\bm{w}_m^\top\right)$, i.e., ${\bm w}_j^\top$ is the $j$-th column of $G$ for $j\in[m]$. It can be seen that $\bm{w}_j\neq \bm{0}$ if and only if $j\in {\rm supp}(W)$, then (i) follows. Moreover, If $\bm{w}_j\neq \bm{0}$, then for any $c\in B$, there are $q^{k-1}$ ${\bm u}$s such that ${\bm u}\bm{w}_j^\top=c$. Therefore, (ii) can be derived easily.

To prove (iii), we write all vectors in ${\bm y}+W$ in a $q^k\times m$ matrix $\mathcal{M}$ where each row is a vector in ${\bm y}+W$. Then ${\bf wt}({\bf y}+W)$ means the weight of $\mathcal{M}$ which can be calculated according to columns.
From (ii) we know, for $j\in {\rm supp}(W)$, the $j$-th column of $\mathcal{M}$ has weight $q^{k-1}(q-1)$; for $j\not\in {\rm supp}(W)$, the $j$-th column of $\mathcal{M}$ has weight $q^k{\bf wt}(y_j)$. Summing up the weight of all $m$ columns, we obtain ${\bf wt}({\bm y}+W)$ as in (iii). Furthermore, set ${\bm y}\in W$, then ${\bf wt}(W)$ is derived. Note it trivially holds $|{\rm supp}(W)|\geq {\rm dim}(W)$.
\end{proof}

Let $\mathcal{C}$ be a linear code of length $n$ defined over $F$ and $\mathcal{B}$ be a fixed basis of $F$ over $B$. We next study the I/O cost with respect to $\mathcal{B}$ of
the repair scheme of node $i^*$ based on the dual codewords  $\{{\bm g}^{(j)}\}_{j=1}^\ell\subseteq \mathcal{C}^\bot$.  Denote by $\hat{\mathcal{B}}$ the dual basis of $\mathcal{B}$. Then, define a matrix
\begin{equation}\label{eq8}
G_{i^*}=\left(\begin{array}{cccc}
       \Phi_{\hat{\mathcal{B}}}({\bm g}_{1}^{(1)})  & \Phi_{\hat{\mathcal{B}}}({\bm g}_{2}^{(1)})&\cdots &\Phi_{\hat{\mathcal{B}}}({\bm g}_{n}^{(1)})  \\

       \Phi_{\hat{\mathcal{B}}}({\bm g}_{1}^{(2)})  & \Phi_{\hat{\mathcal{B}}}({\bm g}_{2}^{(2)})&\cdots &\Phi_{\hat{\mathcal{B}}}({\bm g}_{n}^{(2)})  \\

       \vdots & \vdots & \ddots  & \vdots \\

         \Phi_{\hat{\mathcal{B}}}({\bm g}_{1}^{(\ell)})  & \Phi_{\hat{\mathcal{B}}}({\bm g}_{2}^{(\ell)})&\cdots &\Phi_{\hat{\mathcal{B}}}({\bm g}_{n}^{(\ell)})
    \end{array}\right)\;.\end{equation}
Actually, combining with the definition in (\ref{W_j}) we know $G_{i^*}=(W_1~W_2~\cdots~W_n)$.
Moreover, since ${\rm dim}_B(\{{\bm g}_{i^*}^{(j)}\}_{j=1}^\ell)=\ell$, it follows
$\Phi_{\hat{\mathcal{B}}}({\bm g}_{i^*}^{(1)}),\Phi_{\hat{\mathcal{B}}}({\bm g}_{i^*}^{(2)}),...,\Phi_{\hat{\mathcal{B}}}({\bm g}_{i^*}^{(\ell)})$ are linearly independent over $B$. Let $L_{i^*}$ be the linear space spanned by the rows of $G_{i^*}$ over $B$, i.e.,
\begin{align}\notag
    L_{i^*}&=\{{\bm u}G_{i^*}:{\bm u}\in B^\ell\}\\
    &\label{eq90}={\rm span}_B\big(\{\Phi_{\hat{\mathcal{B}}}({\bm g}^{(1)}),...,\Phi_{\hat{\mathcal{B}}}({\bm g}^{(\ell)})\}\big)\;.
\end{align}
Thus $L_{i^*}$ is an $\ell$-dimensional subspace in $B^{n\ell}$. Furthermore, we obtain the main result in this section.
\begin{theorem}\label{I/O}
Suppose $\{{\bm g}^{(j)}\}_{j=1}^\ell\subseteq \mathcal{C}^\bot$ defines a repair scheme of node $i^*$ in $\mathcal{C}$ and let $L_{i^*}$ be defined as in (\ref{eq90}).  Then the I/O cost of the repair scheme with respect to $\mathcal{B}$ is
    $$\gamma_{I/O}=\frac{{\bf wt}(L_{i^*})}{q^{\ell-1}(q-1)}-\ell.$$
\end{theorem}

\begin{proof}
By Lemma \ref{wt}, we have
\begin{align}\notag
    {\bf wt}(L_{i^*})&=|{\rm supp}(L_{i^*})|q^{\ell-1}(q-1)\\
    &\label{eq9}={\rm nz}(G_{i^*})\cdot q^{\ell-1}(q-1)\;,
\end{align}
where $G_{i^*}$ is defined as in (\ref{eq8}). Moreover,
\begin{equation}
       {\rm nz}(G_{i^*})
       =\sum_{i=1}^n {\rm nz}( W_i)
       =\ell+\!\!\!\sum_{i\in[n]\setminus\left\{\!i^{*}\!\right\}}\!\!\!\mathrm{nz}(W_i)\;,\label{eq11}
\end{equation}
where the notation $W_i$ is defined as in (\ref{W_j}) in Lemma \ref{lem3}, and the second equality in (\ref{eq11}) follows from the fact ${\rm nz}(W_{i^*})=\ell$ due to ${\rm dim}_B(\{{\bm g}_{i^*}^{(j)}\}_{j=1}^\ell)=\ell$.
Finally, using Lemma \ref{lem3} and combining with (\ref{eq9}) and (\ref{eq11}), we have \begin{equation}\label{eq13}\gamma_{I/O}=\!\!\!\sum_{i\in[n]\setminus\left\{\!i^{*}\!\right\}}\!\!\!\mathrm{nz}(W_i)={\rm nz}(G_{i^*})-\ell=\frac{{\bf wt}(L_{i^*})}{q^{\ell-1}(q-1)}-\ell.\end{equation}\end{proof}

\begin{remark}
Note in \cite{fullr=2} and \cite{shortr=2} the I/O cost is calculated according to the first equality in (\ref{eq13}), i.e., by computing $\mathrm{nz}(W_i)$ separately for each $i\in[n]\setminus\{i^*\}$. In contrast, we take the failed node $i^*$ also into the calculation by investigating the linear space $L_{i^*}$. Note  $L_{i^*}$ is actually a $B$-linear subcode of $\mathcal{C}^\bot$ representing each coordinate as a vector in $B^\ell$ through the map $\Phi_{\hat{\mathcal{B}}}$. When $\mathcal{C}$ is an RS code, its dual code is well known, leading to an easier calculation of the I/O cost.
\end{remark}

\section{Lower Bound on The I/O Cost for RS Codes with Two and Three Parities}\label{Sec4}

Using the formula obtained in Theorem \ref{I/O}, in this section we establish lower bounds on the I/O cost for repairing full-length RS codes with two and three parities. Our bound for RS codes with two parities extends the bound derived in \cite{fullr=2} from fields of characteristic $2$ to any finite field. Moreover, the bound is tight due to a construction given later in Section \ref{Sec5}. However, our bound for RS codes with three parities only applies to fields of characteristic $2$ because we rely on the linearity of polynomials of degree $2$.

Let $\mathcal{A}=\left\{\alpha_1, \alpha_2 ,...,\alpha_n \right\} \subseteq F $ be the set of evaluation points. Then the RS code $RS(\mathcal{A},k)$ over $F$ is defined by
\begin{equation*}\label{RS}
RS(\mathcal{A},k)=\{(f(\alpha_1),...,f(\alpha_n)):f\in F[x], \mathrm{deg}(f)\leq k-1\}.
\end{equation*}
When $\mathcal{A}=F$, the RS code is called full-length. Moreover, the dual of a full-length RS code is still an RS code, i.e., $RS(F,k)^\bot=RS(F,q^\ell-k)$.

\begin{theorem}\label{r=2}
For $RS(F,q^\ell-2)$ over $F$, the I/O cost of any linear repair scheme satisfies:
      \begin{equation*}
        \gamma_{I/O} \geq (n-1)\ell-q^{\ell-1},
      \end{equation*}where $n=q^\ell$.
\end{theorem}

\begin{proof}
For simplicity, denote $F=\{\alpha_1\!=\!0,\alpha_2,...,\alpha_n\}$.
By \citep[Lemma 8]{IOfulllength}, the repair scheme of full-length RS codes has the same repair bandwidth and I/O cost at all nodes, so it suffices to just examine the repair scheme of node $0$ (corresponding to the evaluation point $\alpha_1=0$).

Let $\{{\bm g}^{(j)}\}_{j=1}^\ell$ be the $\ell$ dual codewords that define a repair scheme of node 0. Since $RS(F,q^\ell-2)^\bot=RS(F,2)$, we have ${\bm g}^{(j)}\!=(g_j(\alpha_1),...,g_j(\alpha_n))$, where $g_j(x)=\lambda_jx+\mu_j\in F[x]$ for some $\lambda_j,\mu_j\in F$. Moreover, the condition ${\rm dim}_B(\{g_{j}(0)\}_{j=1}^\ell)=\ell$ implies that ${\rm dim}_{B}(\{\mu_j\}_{j=1}^{\ell})=\ell$. To compute the I/O cost, by Theorem \ref{I/O} we need to compute ${\bf wt}(L_0)$, where
$L_0=\big\{\Phi_{\hat{\mathcal{B}}}(\sum_{j=1}^\ell u_j{\bm g}^{(j)}):{\bm u}=(u_1,...,u_\ell)\in B^\ell\big\}$.
For simplicity, define $g_{\bm u}(x)\!=\!\sum_{j=1}^\ell u_jg_j(x)$ and ${\bm g}_{\bm u}=(g_{\bm u}(\alpha_1),...,g_{\bm u}(\alpha_n))$. Then it follows
$$\textstyle{\bf wt}(L_0)=\sum_{{\bm u}\in B^\ell}{\bf wt}({\Phi}_{\hat{\mathcal{B}}}({\bm g}_{\bm u})).$$

Assume ${\rm dim}_{B}(\{\lambda_j\}_{j=1}^{\ell})=m$ and denote
$W=\{{\bm u}\in B^\ell:\sum_{j=1}^\ell u_j\lambda_j=0\}$. Then $\mathrm{dim}(W)=\ell-m$.  Next we compute ${\bf wt}({\Phi}_{\hat{\mathcal{B}}}({\bm g}_{\bm u}))$ with respect to ${\bm u}\in W$ and ${\bm u}\not\in W$, respectively.
\begin{enumerate}
\item[(1)]${\bm u}\in W$.

Then $g_{\bm u}(x)\!=\!\sum_{j=1}^\ell u_j\mu_j$ is a constant.
Define a map $\sigma$ from $B^\ell$ to $B^\ell$ such that  $\sigma(\bm u)=\Phi_{\hat{\mathcal{B}}}(\sum_{j=1}^\ell u_j\mu_j)$ for all ${\bm u}\in B^\ell$. It follows
${\bf wt}({\Phi}_{\hat{\mathcal{B}}}({\bm g}_{\bm u}))=q^\ell{\bf wt}(\sigma({\bm u}))$.

\item[(2)] ${\bm u}\not\in W$.

Then $g_{{\bm u}}(x)=\sum_{j=1}^\ell u_j\lambda_jx+\sum_{j=1}^\ell u_j\mu_j$ is a polynomial of degree one which is also a  bijection from $F$ to $F$, so ${\bm g}_{\bm u}$ is a vector in $F^n$ with all coordinates forming a permutation of $F$. Recall that ${\Phi}_{\hat{\mathcal{B}}}$ is a bijection from $F$ to $B^\ell$. Thus ${\bf wt}({\Phi}_{\hat{\mathcal{B}}}({\bm g}_{\bm u}))={\bf wt}(B^\ell)=\ell q^{\ell-1}(q-1)$, where the last equality is from Lemma \ref{wt}.
\end{enumerate}

Finally,
\fontsize{9.5pt}{1pt}
{\begin{align}
{\bf wt}(L_0)&=\sum_{{\bm u}\in W}{\bf wt}({\Phi}_{\hat{\mathcal{B}}}({\bm g}_{\bm u}))+\sum_{{\bm u}\in B^\ell\setminus W}{\bf wt}({\Phi}_{\hat{\mathcal{B}}}({\bm g}_{\bm u}))\notag\\
&=q^\ell\sum_{{\bm u}\in W}{\bf wt}(\sigma(\bm u))+|B^\ell\setminus W|\cdot \ell q^{\ell-1}(q-1)\notag\\
&=q^\ell{\bf wt}(\sigma(W))+\ell q^{\ell-1}(q-1)(q^\ell-q^{\ell-m})\label{eq14}\\
&\geq (\ell-m)q^{2\ell-m-1}(q-1)+\ell q^{\ell-1}(q-1)(q^\ell-q^{\ell-m})\;,\label{eq15}
\end{align}}

\noindent where (\ref{eq14}) is because $\sigma$ is a bijection due to ${\rm dim}_{B}(\{\mu_j\}_{j=1}^{\ell})=\ell$. Moreover, it is easy to see that $\sigma(W)$ is a $B$-linear space of dimension $\ell-m$. Then (\ref{eq15}) comes from Lemma \ref{wt}. By Theorem \ref{I/O}, we obtain
\begin{equation}\notag
  \gamma_{I/O} =\frac{{\bf wt}(L_0)}{q^{\ell-1}(q-1)}-\ell\geq \ell q^\ell-mq^{\ell-m}-\ell\geq \ell(q^\ell-1)-q^{\ell-1},
  \end{equation}
where the last inequality follows as $mq^{\ell-m}$ is maximum at $m=1$.
\end{proof}
Note when $q=2$, the lower bound obtained in Theorem \ref{r=2} coincides with the bound derived in \cite{fullr=2}. However, our computation is quite simple because we consider the weight of ${\bm g}_{\bm u}$ whose coordinates actually form the value table of the polynomial $g_{\bm u}(x)$. In a similar way, we continue to derive a lower bound on the I/O cost for full-length RS codes with three parities. Since in this case the dual codewords correspond to polynomials of degree at most $2$ and a polynomial of degree $2$ is a (affine) $q$-polynomial over fields of characteristic $2$, we restrict to $B=\mathbb{F}_2$ for easier computation.
\begin{theorem}\label{r=3}
For $RS\left(\mathbb{F}_{2^\ell},2^\ell-3\right)$ over $\mathbb{F}_{2^\ell}$, the I/O cost of any linear repair scheme satisfies:
      \begin{equation*}
        \gamma_{I/O}\geq (n-1)\ell-n-2^{\ell-3},
      \end{equation*}
where $n=2^\ell$.
\end{theorem}

The proof is presented in Appendix \ref{appendixA}  and a repair scheme with $\gamma_{I/O}=(n-1)\ell-n$ is given in Section \ref{Sec5}.

\section{Repairing full-length RS codes via affine $q$-polynomials}\label{Sec5}

In this section, we build a repair scheme for the full-length RS code $RS(F,k)$. Suppose $n=|F|=q^\ell$ and $n-k\geq q^s+1$ for some $0\leq s<\ell$. Denote $F=\{\alpha_1\!=\!0,\alpha_2,...,\alpha_n\}$. The repair scheme is defined by $\ell$ dual codewords $\{g_j(x)\}_{j=1}^\ell$ which are specially designed affine $q$-polynomials. Then we compute the I/O cost of the repair scheme with respect to the basis $\mathcal{B}=\{\beta^{(1)},...,\beta^{(\ell)}\}$. It turns out that when $n-k=2$ the repair scheme achieves the optimal I/O cost and when $n-k=3$ its I/O cost is lower than previous schemes.

First we recall some basics about $q$-polynomials. A $q$-polynomial over $F$ has the form $L(x)=\sum_{i=0}^m\lambda_ix^{q^i}$, where $\lambda_i\in F$. It is known each $q$-polynomial induces a $B$-linear operator on any extension field of $F$.
Hence, ${\rm Im}(L)$ is a subspace of $F$ and ${\rm dim}_B({\rm Im}(L))\geq \ell-m$. As an example, the trace function ${\rm Tr}(\cdot)$ of $F$ over $B$ is a $q$-polynomial. Denote $K={\rm Ker}({\rm Tr})$. We know ${\rm dim}_B(K)=\ell-1$. Suppose $L(x)$ is a $q$-polynomial over $F$ and $\alpha\in F$. Then the polynomial $L(x)+\alpha$ is called an affine $q$-polynomial. Next, we present a result which will be used to select the $q$-polynomials in the repair scheme.

\begin{theorem}\label{q-poly}
 Assume $1\leq t<\ell$ and $\beta_1,...,\beta_t\in F$ are linearly independent over $B$. Let $(\theta_0,\theta_1,...,\theta_t)$ with $\theta_t\neq 0$ be a solution to the system
\begin{equation}\label{cond}
        \left( {\begin{array}{cccc} \beta_1 & \beta_1^{q} &  \cdots &\beta_1^{q^t}\\
\beta_2 & \beta_2^{q} &  \cdots &\beta_2^{q^t}\\
\vdots & \vdots & \ddots & \vdots \\
 \beta_t & \beta_t^{q} &  \cdots &\beta_t^{q^t}\end {array}}\right) \left({\begin{array}{c}
      \theta_t  \\
      \theta_{t-1}^q\\
      \vdots\\
      \theta_0^{q^t}
 \end{array}}\right)={\bm 0}.
    \end{equation}
Define $L(x)=\sum_{j=0}^t \theta_j x^{q^{j}}$. Then  ${\rm Im}(L)=\bigcap_{i=1}^t \beta_i^{-1}K$.
\end{theorem}
We first show the solution $(\theta_0,...,\theta_t)\in F^{t+1}$ with $\theta_t\neq 0$ does exist. Clearly, the coefficient matrix on the left side of (\ref{cond}) is a $t\times(t+1)$ matrix of which the last $t$ columns are linearly independent. Thus $\theta_t=0$ always implies $\theta_i=0$ for all $i\in[0,t]$.  Since the system must have nonzero solutions and the map $x\mapsto x^{q^i}$ is an automorphism of $F$ over $B$, we have the solution $(\theta_0,...,\theta_t)$ with $\theta_t\neq 0$. The remaining proof of Theorem \ref{q-poly} is presented in Appendix \ref{appendixB}.

Moreover, we have following lemma from \cite{IOfulllength} to be used in later proof.
\begin{lemma}\label{lem12}(\cite{IOfulllength}) If $\{\beta_j\}_{j=1}^t\subseteq F$ are linearly independent over $B$, then $\mathrm{dim}_B\big(\bigcap_{j=1}^t\beta_{j}^{-1}K\big)=\ell-t$\;.
\end{lemma}

Now we are ready to present the repair scheme. As before, let $\hat{\mathcal{B}}=\{\gamma^{(1)},...,\gamma^{(\ell)}\}$ be the dual basis of $\mathcal{B}$.

\begin{construction}\label{cons1}
For $j\in[s+1]$, let $L_j(x)$ be a $q$-polynomial such that ${\rm Im}(L_j)=\bigcap_{i\in[s+1]\setminus\{j\}}(\beta^{(i)})^{-1}K$. Specifically, the coefficients of $L_j(x)$ can be chosen as a nonzero solution to a system similar to (\ref{cond}) replacing $t$ with $s$ and $\beta_1,...,\beta_t$ with the $s$ elements in $\{\beta^{(i)}:i\in[s+1],i\neq j\}$. Then define
$$ g_j(x)=\begin{cases}
L_j(x)+\gamma^{(j)},&j\in[s+1],\\
\gamma^{(j)},	&j\in[s+2,\ell].	\\
\end{cases}$$
\end{construction}

We first check the polynomials $\{g_j(x)\}_{j=1}^\ell$ define a repair scheme of node $\alpha_1=0$ for $RS(F,k)$. Since ${\rm deg}(g_j(x))=q^s<n-k$ for all $j\in[\ell]$, these polynomials correspond to codewords in the dual code. Moreover, $g_j(0)=\gamma^{(j)}$, so it obviously holds ${\rm dim}_B(\{g_j(0)\}_{j=1}^\ell)=\ell$. In the following theorem, we compute the I/O cost and bandwidth of the repair scheme. 
\begin{theorem}\label{gamma_I/O}
Let $n=q^{\ell}$ and $n-k\geq q^s+1$.  The I/O cost with respect to $\mathcal{B}$ of the repair scheme given in Construction \ref{cons1} for $RS(F,k)$ is
\begin{equation}\label{12}
    \gamma_{I/O}=(n-1)\ell-(s+1)\cdot q^{\ell-1}.
\end{equation}
Moreover, the repair bandwidth of this scheme equals its I/O cost.
\end{theorem}

\begin{proof}
According to Lemma \ref{lem3}, $\gamma_{I/O}=\sum_{i=2}^n\mathrm{nz}(W_i)$, where 
\begin{equation}\label{W_i}
	W_i=\left( {\setlength{\arraycolsep}{1pt}\begin{array}{cccc}\mathrm{Tr}(g_1(\alpha_i)\beta^{(1)}) & \mathrm{Tr}(g_1(\alpha_i)\beta^{(2)}) &  \cdots &\mathrm{Tr}(g_1(\alpha_i)\beta^{(\ell)})\\
			\mathrm{Tr}(g_2(\alpha_i)\beta^{(1)}) & \mathrm{Tr}(g_2(\alpha_i)\beta^{(2)}) &  \cdots &\mathrm{Tr}(g_2(\alpha_i)\beta^{(\ell)})\\
			\vdots & \vdots & \ddots & \vdots \\
			\mathrm{Tr}(g_\ell(\alpha_i)\beta^{(1)}) & \mathrm{Tr}(g_\ell(\alpha_i)\beta^{(2)}) &  \cdots &\mathrm{Tr}(g_\ell(\alpha_i)\beta^{(\ell)})\end {array}}\right).
	\end{equation}
	 To simplify the calculation, in the following we compute $\sum_{i\in[n]}{\rm nz}(W_i)$ instead of $\sum_{i=2}^n{\rm nz}(W_i)$. Due to the definition of $\{g_j(x)\}_{j=1}^\ell$ in Construction \ref{cons1}, we have for $j\in[s+1]$,
	 \begin{align}{\rm Tr}(g_j(\alpha_i)\beta^{(t)})&={\rm Tr}(L_j(\alpha_i)\beta^{(t)})+{\rm Tr}(\gamma^{(j)}\beta^{(t)})\notag\\
	 	&=\begin{cases}
	 		{\rm Tr}(L_j(\alpha_i)\beta^{(j)})+1,&t=j\\
	 		0,&t\in[s+1]\setminus\{j\}\\
	 		{\rm Tr}(L_j(\alpha_i)\beta^{(t)}),&t\in[s+2,\ell]
	 	\end{cases},\label{eq18}\end{align}
	 while for $j\in[s+2,\ell]$,
	 \begin{align}{\rm Tr}(g_j(\alpha_i)\beta^{(t)})&={\rm Tr}(\gamma^{(j)}\beta^{(t)})\notag\\
	 	&=\begin{cases}
	 		1,&~~t=j\\
	 		0,&~~t\in[\ell]\setminus\{j\}
	 	\end{cases}\;.\label{eq19}\end{align}
	 Note both in (\ref{eq18}) and (\ref{eq19}) we use the fact that $\{\gamma^{(1)},...,\gamma^{(\ell)}\}$ is the dual basis of $\mathcal{B}$. Moreover, in the second case of (\ref{eq18}) we also use the condition ${\rm Im}(L_j)=\bigcap_{i\in[s+1]\setminus\{j\}}(\beta^{(i)})^{-1}K$. In summary, $W_i$ for $i\in[n]$ has the following form
	 \begin{equation}\label{eq20}
	 	W_i=\left( {\begin{array}{ccc|ccc}
	 			\omega_{i,1} &\cdots &0 & * &\cdots &*\\
	 			\vdots &\ddots &\vdots & \vdots &\ddots & \vdots \\
	 			0 &\cdots &\omega_{i,s+1} & * &\cdots &*\\ \hline
	 			0 &\cdots &0 & 1 &\cdots &0\\
	 			\vdots &\ddots &\vdots & \vdots &\ddots & \vdots \\
	 			0 &\cdots &0 & 0 &\cdots &1\\
	 			\end {array}}\right),
	 	\end{equation}where $\omega_{i,j}={\rm Tr}(L_j(\alpha_i)\beta^{(j)})+1$ for $j\in[s+1]$. Note the upper left corner of $W_i$ is a $(s+1)\times (s+1)$ diagonal matrix, the lower left corner is an all-zero matrix, and the lower right corner is an identity matrix. It immediately follows ${\rm rank}(W_i)={\rm nz}(W_i)$, so the repair bandwidth and I/O cost are equal.
	 	Denote the $(s+1)\times (s+1)$ diagonal matrix in the upper left corner of $W_i$ by $W_i'$. It has ${\rm nz}(W_i)={\rm nz}(W_i')+\ell-s-1$. Then we compute
	 	\begin{align}\notag
	 		\sum_{i\in[n]}{\rm nz}(W_i')&=\sum_{i\in[n]}\big(s+1-|\{j\in[s+1]:\omega_{i,j}=0\}|\big)\\
	 		&\label{r3}=n(s+1)-\sum_{j\in[s+1]}|\{i\in[n]:\omega_{i,j}=0\}|\\&\notag=n(s+1)\\&\notag\quad-\sum_{j\in[s+1]}|\{i\in[n]:{\rm Tr}(L_j(\alpha_i)\beta^{(j)})=-1\}|\;,
	 	\end{align}
	 	where (\ref{r3}) uses double counting.
	 	
	 	Since ${\rm Im}(L_j)=\bigcap_{i\in[s+1]\setminus\{j\}}(\beta^{(i)})^{-1}K$, by Lemma \ref{lem12} it has ${\rm dim}_B({\rm Im}(L_j))=\ell-s$, and thus ${\rm dim}_B({\rm Ker}(L_j))=s$. Also,
	 	$\Phi_{\hat{\mathcal{B}}}({\rm Im}(L_j))$ is an $(\ell-s)$-dimensional subspace of $B^\ell$. Moreover, because $L_j(\alpha_1)=\gamma^{(j)}$ and $\Phi_{\hat{\mathcal{B}}}(\gamma^{(j)})={\bm e}_j$, we know $j\in{\rm supp}\big(\Phi_{\hat{\mathcal{B}}}({\rm Im}(L_j))\big)$. Then by Lemma \ref{wt}, there are $q^{\ell-s-1}$ $\alpha$s in ${\rm Im}(L_j)$ such that the $j$-th position of $\Phi_{\hat{\mathcal{B}}}(\alpha)$ is $-1$, i.e., ${\rm Tr}(\alpha\beta^{(j)})=-1$. Moreover, for each such $\alpha$, $L_j^{-1}(\alpha)=q^s$. Therefore,
	 	$$|\{i\in[n]:{\rm Tr}(L_j(\alpha_i)\beta^{(j)})=-1\}|=q^s\cdot q^{\ell-s-1}=q^{\ell-1}\;,$$
	 	and thus $\sum_{i\in[n]}{\rm nz}(W_i')=(s+1)(n-q^{\ell-1})$. Finally, it follows
	 	\begin{align*}
	 		\sum_{i\in[n]}{\rm nz}(W_i)&=\sum_{i\in[n]}({\rm nz}(W_i')+\ell-s-1)\\
	 		&=\ell q^{\ell}-(s+1)q^{\ell-1}\;.
	 	\end{align*}
	 	By the definition in (\ref{W_i}), we know $W_1$ is the identity matrix $I_{\ell}$, so $\gamma_{I/O}=\sum_{i\in[n]}{\rm nz}(W_i)-\ell$ which completes the proof.
\end{proof}
Although the formula derived in Theorem \ref{I/O} was not used in the above calculation of the I/O cost, it actually inspires the determination of the affine $q$-polynomials in Construction \ref{cons1} for lower I/O cost.

\begin{remark}
The repair schemes given in \cite{RSrepair} and \cite{obRS} for the full-length RS code $RS(F,k)$ with $n-k\geq q^s$ reach the optimal repair bandwidth $(n-1)(\ell-s)$. However, Dau et al. \cite{IOfulllength} showed the schemes incur an I/O cost of $(n-q^s)\ell\geq k\ell$, while $k\ell$ means the I/O cost of the trivial repair approach (i.e., by connecting to $k$ nodes and downloading all data stored on the nodes). In contrast, our schemes incur an I/O cost (also repair bandwidth) of $(n-1)\ell-(s+1)q^{\ell-1}$ which is lower than $k\ell$ when $n\!-\!k\leq \frac{(s+1)q^{\ell-1}}{\ell}$. Suppose $q^s\!<\! n\!-\!k\!\leq\! q^{s+1}$. Then the condition $n\!-\!k\leq \frac{(s+1)q^{\ell-1}}{\ell}$ always holds for $\ell\geq s+2+\log_q\frac{\ell}{s+1}$ which is true in most cases. Moreover, when $q=2$ and $\ell\leq 2^{\ell-s}$, the sum of repair bandwidth and I/O cost of the schemes proposed in \cite{RSrepair} and \cite{obRS} is higher than that of our scheme, i.e.,
\begin{align*}
&\big((n-1)(\ell-s)+(n-2^s)\ell\big)-2\big((n-1)\ell-(s+1)2^{\ell-1}\big)\\
&=\ell+s+(2^{\ell-s}-\ell)2^s
\end{align*}
which means our scheme achieves a more significant reduction in I/O cost with a small bandwidth sacrifice.

\end{remark}

\begin{corollary}
Set $n-k=2$ and $s=0$. Then Construction \ref{cons1}  provides a repair scheme for $RS(F,k)$ with the I/O cost  $(n-1)\ell-q^{\ell-1}$ which is optimal with respect to the bound proved in Theorem \ref{r=2}.
\end{corollary}
Moreover, our scheme for three parities only takes $2^{\ell-3}$ bits away from the lower bound in Theorem \ref{r=3}:
\begin{corollary}
Set $q=2$, $n-k=3$ and $s=1$. Then Construction \ref{cons1}  provides a repair scheme for $RS(F,k)$ with the I/O cost  $(n-1)\ell-n$.
\end{corollary}

\section{Conclusion}\label{Sec6}
In this paper, we provide a general formula for computing the I/O cost of linear repair schemes. It helps us establish lower bounds on the I/O cost of repair schemes for full-length RS codes and construct repair schemes with lower I/O cost. However, we only discuss full-length RS codes with two and three parities in this work. In the future work, we expect to derive lower bounds on the I/O cost and provide linear repair schemes for general RS Codes using the formula.
\appendices
\section{ proof of Theorem \ref{r=3}}\label{appendixA}
\begin{proof}
	Let $\{{\bm g}^{(j)}\}_{j=1}^\ell$ be the $\ell$ dual codewords that correspond to the repair scheme of node 0. Since $RS(\mathbb{F}_{2^\ell},2^\ell-3)^\bot=RS(\mathbb{F}_{2^\ell},3)$, we have ${\bm g}^{(j)}\!=\!(g_j(\alpha_1),...,g_j(\alpha_n))$, where $g_j(x)=\lambda_jx^2+\eta_jx+\mu_j$ for some $\lambda_j,\eta_j,\mu_j\in \mathbb{F}_{2^\ell}$. Let $g_{\bm u}(x), {\bm g}_{\bm u}$ and $\sigma$ be defined as in Theorem 6. %\ref{r=2}
	Then we proceed to compute ${\bf wt}(L_0)$, where $L_0=\{{\Phi}_{\hat{\mathcal{B}}}({\bm g}_{\bm u}):{\bm u}\in \mathbb{F}_2^\ell\}$.
	
	Denote $W\!=\!\{(\sum_{j=1}^\ell u_j\lambda_j,\sum_{j=1}^{\ell}u_j\eta_j): {\bm u}\in \mathbb{F}_2^\ell\}$. For any $(\lambda,\eta)\in W$, define
	\begin{equation*}
		U_{(\lambda,\eta)}=\{{\bm u}\in \mathbb{F}_2^{\ell}:
		\sum_{j=1}^{\ell} u_j\lambda_j=\lambda, \sum_{j=1}^{\ell}u_j\eta_j=\eta\}\;.
	\end{equation*}
	Obviously, $U_{(0,0)}$ is a linear subspace of $\mathbb{F}_2^\ell$ and $U_{(\lambda,\eta)}$ is a coset of $U_{(0,0)}$ for any $(\lambda,\eta)\in W$.
	Assume ${\rm dim}(U_{(0,0)})=t$. Then, ${\rm dim}_{\mathbb{F}_2}(W)\!=\!\ell-t$.
	Since $\{U_{(\lambda,\eta)}\}_{(\lambda,\eta)\in W}$ forms a partition of $\mathbb{F}_2^\ell$, it follows $${\bf wt}(L_0)=\!\!\!\sum_{(\lambda,\eta)\in W}\sum_{{\bm u}\in U_{(\lambda,\eta)}}\!\!{\bf wt}({\Phi}_{\hat{\mathcal{B}}}({\bm g}_{\bm u}))\;.$$
	
	Next we  explore ${\bf wt}({\Phi}_{\hat{\mathcal{B}}}({\bm g}_{\bm u}))$ for any ${\bm u}\in U_{(\lambda,\eta)}$. It can be seen
	\begin{equation}{\bf wt}({\Phi}_{\hat{\mathcal{B}}}({\bm g}_{\bm u}))=\sum_{\alpha\in\mathbb{F}_2^\ell}{\bf wt}(\Phi_{\hat{\mathcal{B}}}(\lambda\alpha^2+\eta\alpha)+\sigma({\bm u}))\;.\label{eq25}\end{equation}  There are three cases:
	\begin{itemize}
		\item $\lambda=\eta=0$. Then from (\ref{eq25}) it has ${\bf wt}({\Phi}_{\hat{\mathcal{B}}}({\bm g}_{\bm u}))=2^\ell{\bf wt}(\sigma({\bm u}))$. Furthermore, we have
		\fontsize{9.5pt}{1pt}
		{\begin{equation}\sum_{{\bm u}\in U_{(0,0)}}\!\!\!\!\!{\bf wt}({\Phi}_{\hat{\mathcal{B}}}({\bm g}_{\bm u}))=2^\ell\!\!\!\!\!\sum_{{\bm u}\in U_{(0,0)}}\!\!\!\!\!{\bf wt}(\sigma({\bm u}))=2^\ell{\bf wt}(\sigma(U_{(0,0)})).\label{eq26}\end{equation}}
		
		\item Only one of $\lambda$ and $\eta$ is $0$. Then $\lambda x^2+\eta x$ is a bijection from $\mathbb{F}_{2^\ell}$ to $\mathbb{F}_{2^\ell}$. Then from (\ref{eq25}) it has ${\bf wt}({\Phi}_{\hat{\mathcal{B}}}({\bm g}_{\bm u}))={\bf wt}(\mathbb{F}_2^\ell)=2^{\ell-1}\ell$;
		\item $\lambda\neq 0$ and $\eta\neq 0$.  Then $\lambda x^2+\eta x$ is a $q$-polynomial over $\mathbb{F}_{2^\ell}$ with the kernel $\{0,-\lambda^{-1}\eta\}$. Thus for any $\alpha\in\mathbb{F}_{2^\ell}$, both $\alpha$ and $\alpha-\lambda^{-1}\eta$ map to the same value under this polynomial. For simplicity, denote $V_{(\lambda,\eta)}=\{\Phi_{\hat{\mathcal{B}}}(\lambda\alpha^2+\eta\alpha):\alpha\in \mathbb{F}_{2^\ell}\}$. It follows $V_{(\lambda,\eta)}$ is an $(\ell-1)$-dimensional subspace of $\mathbb{F}_2^\ell$.  Then combining with (\ref{eq25}) and Lemma 4, we have
		\begin{align}{\bf wt}({\Phi}_{\hat{\mathcal{B}}}({\bm g}_{\bm u}))&=2{\bf wt}(V_{(\lambda,\eta)}+\sigma({\bm u}))\notag\\&\notag=2^{\ell-1}|{\rm supp}(V_{(\lambda,\eta)})|\\&\quad+2^\ell|{\rm supp}(\sigma({\bm u}))\setminus{\rm supp}(V_{(\lambda,\eta)})|\label{eq27}\\
			&\geq 2^{\ell-1}(\ell-1)\;,\label{eq28}\end{align}where the inequality (\ref{eq28}) comes from $|{\rm supp}(V_{(\lambda,\eta)})|\geq {\rm dim}(V_{(\lambda,\eta)})=\ell-1$ and $|{\rm supp}(\sigma({\bm u}))\setminus{\rm supp}(V_{(\lambda,\eta)})|\geq0$.
	\end{itemize}
	
	Actually, it is easy to check that ${\bf wt}({\Phi}_{\hat{\mathcal{B}}}({\bm g}_{\bm u}))$ in the first two cases also satisfy the equality (\ref{eq27}). To derive a tighter lower bound on ${\bf wt}(L_0)$, one should give more precise estimate on the ${\bf wt}({\Phi}_{\hat{\mathcal{B}}}({\bm g}_{\bm u}))$ that may be higher than the rough bound in (\ref{eq28}). To this end, we specially choose ${\bm u}'\in \mathbb{F}_2^\ell$ such that $\sigma({\bm u}')\!=\!{\bm 1}$. Note such a ${\bm u}'$ definitely exists because $\sigma$ is a bijection from $\mathbb{F}_2^\ell$ to $\mathbb{F}_2^\ell$. Suppose ${\bm u}'\in U_{(\lambda',\eta')}$ for some $(\lambda',\eta')\in W$, then $|{\rm supp}(\sigma(U_{(\lambda',\eta')})|=\ell$. Next we estimate ${\bf wt}(L_0)$ in two cases.
	
	\begin{enumerate}
		\item If $(\lambda',\eta')=(0,0)$, then
		\begin{align}
			{\bf wt}(L_0)&=\sum_{{\bm u}\in U_{(0,0)}}{\bf wt}({\Phi}_{\hat{\mathcal{B}}}({\bm g}_{\bm u}))+\!\!\!\sum_{{\bm u}\in \mathbb{F}_2^\ell\setminus U_{(0,0)}}\!\!\!{\bf wt}({\Phi}_{\hat{\mathcal{B}}}({\bm g}_{\bm u}))\notag\\
			&\geq2^\ell{\bf wt}(\sigma(U_{(0,0)}))+|\mathbb{F}_2^\ell\setminus U_{(0,0)}|\cdot 2^{\ell-1}(\ell-1)\label{eq29}\\
			&=2^\ell\cdot\ell\cdot 2^{t-1}+(2^\ell-2^t)2^{\ell-1}(\ell-1)\label{eq30}\\
			&=2^{2\ell-1}(\ell-1)+2^{\ell+t-1}\notag\\
			&\geq 2^{2\ell-1}(\ell-1)+2^{\ell-1}\label{eq31}
		\end{align}
		where (\ref{eq29}) is from (\ref{eq26}) and (\ref{eq28}), and the equality (\ref{eq30}) is based on Lemma 4 combining with the facts $|{\rm supp}(\sigma(U_{(0,0)}))|\!=\!\ell$ and ${\dim}(U_{(0,0)})\!=\!t$. Note ${\bm u}'\in U_{(0,0)}$ in this case, so $|{\rm supp}(\sigma(U_{(0,0)}))|\!=\!\ell$. The inequality (\ref{eq31}) holds because $t\geq 0$ \footnote{Set $\lambda_i=\eta_i$ for $i\in[\ell]$ and $\lambda_1,...,\lambda_\ell$ be linearly independent over $\mathbb{F}_2$. Then it has ${\dim}(U_{(0,0)})=t=0$. In a similar setting, one can see that $t$ can range from $0$ to $\ell-1$.}.
		
		\item If $(\lambda',\eta')\neq(0,0)$, then \begin{align}{\bf wt}(L_0)&\notag=\sum_{{\bm u}\in U_{(0,0)}}{\bf wt}({\Phi}_{\hat{\mathcal{B}}}({\bm g}_{\bm u}))+\!\!\!\sum_{{\bm u}\in U_{(\lambda',\eta')}}\!\!\!{\bf wt}({\Phi}_{\hat{\mathcal{B}}}({\bm g}_{\bm u}))\\
			&\quad+\!\!\!\sum_{{\bm u}\not\in U_{(0,0)}\cup U_{(\lambda',\eta')}}\!\!\!{\bf wt}({\Phi}_{\hat{\mathcal{B}}}({\bm g}_{\bm u}))\notag\\
			&\geq t2^{\ell+t-1}+\!\!\!\sum_{{\bm u}\in U_{(\lambda',\eta')}}\!\!\!{\bf wt}({\Phi}_{\hat{\mathcal{B}}}({\bm g}_{\bm u}))\notag\\
			&\quad+(2^\ell-2^{t+1})2^{\ell-1}(\ell-1)\label{eq32}\end{align}
		where (\ref{eq32}) follows from the same deduction as that used in (\ref{eq30}) except that we only have $|{\rm supp}(U_{(0,0)})|\geq {\rm dim}(U_{(0,0)})\!=\!t$ in this case. So we are left to estimate the second item in (\ref{eq32}). Recall $|{\rm supp}(V_{(\lambda',\eta')})|\geq {\rm dim}(V_{(\lambda',\eta')})\!\geq\!\ell-1$ in this case. Then we have the following discussions.
		\begin{itemize}
			\item If $|{\rm supp}(V_{(\lambda',\eta')})|=\ell$, then for any ${\bm u}\in U_{(\lambda',\eta')}$ we have
			${\bf wt}({\Phi}_{\hat{\mathcal{B}}}({\bm g}_{\bm u}))=2^{\ell-1}\ell$ according to (\ref{eq27}). As a result, $\sum_{{\bm u}\in U_{(\lambda',\eta')}}{\bf wt}({\Phi}_{\hat{\mathcal{B}}}({\bm g}_{\bm u}))=2^{\ell+t-1}\ell$.
			\item If $|{\rm supp}(V_{(\lambda',\eta')})|=\ell-1$, then there exists $i^*\in[\ell]$ such that $i^*\in {\rm supp}(\sigma(U_{(\lambda',\eta')}))\setminus {\rm supp}(V_{(\lambda',\eta')})$ due to $|{\rm supp}(\sigma(U_{(\lambda',\eta')})|=\ell$. Since $\sigma(U_{(\lambda',\eta')})$ is a coset of $\sigma(U_{(0,0)})$, by Lemma 4 we know there are at least $2^{t-1}$ elements in $\sigma(U_{(\lambda',\eta')})$ that contain $i^*$ in the support set. So by (\ref{eq27}) it has
			\begin{align*}
				\sum_{{\bm u}\in U_{(\lambda',\eta')}}\!\!\!\!{\bf wt}({\Phi}_{\hat{\mathcal{B}}}({\bm g}_{\bm u}))&\geq 2^t\cdot2^{\ell-1}(\ell-1)+2^{t-1}\cdot 2^{\ell}\\&=2^{\ell+t-1}\ell\;.           \end{align*}
		\end{itemize}
		That is, we always have $\sum_{{\bm u}\in U_{(\lambda',\eta')}}{\bf wt}({\Phi}_{\hat{\mathcal{B}}}({\bm g}_{\bm u}))\!\!\!\geq 2^{\ell+t-1}\ell$ in this case. Then combining with (\ref{eq32}) it follows
		\begin{align}
			{\bf wt}(L_0)&\geq t2^{\ell+t-1}+2^{\ell+t-1}\ell+(2^\ell-2^{t+1})2^{\ell-1}(\ell-1)\notag\\
			&=2^{2\ell-1}(\ell-1)+2^{\ell+t-1}(t+2-\ell)\notag\\
			&\geq 2^{2\ell-1}(\ell-1)-2^{2\ell-4}\label{eq33}
		\end{align}
		where the inequality (\ref{eq33}) is due to the fact that $2^{\ell+t-1}(t+2-\ell)$ reaches the minimum at $t=\ell-3$ as $t$ ranges from $0$ to $\ell-1$. Note $\ell\geq 3$ always holds for nontrivial $RS(\mathbb{F}_{2^\ell}, 2^\ell-3)$.
	\end{enumerate}
	
	Finally, combining (\ref{eq31}) and (\ref{eq33}), we conclude ${\bf wt}(L_0)\geq 2^{2\ell-1}(\ell-1)-2^{2\ell-4}$ and then the lower bound on $\gamma_{I/O}$ follows from Theorem 5.%\ref{I/O}.
\end{proof}
\section{proof of Theorem \ref{q-poly}}\label{appendixB}
\begin{proof}
	We first prove a related claim: for any $\beta\in F$ and $\beta\neq 0$, if
	$\sum_{j=0}^t(\beta\theta_j)^{q^{t-j}}=0$, then ${\rm Im}(L)\subseteq\beta^{-1}K$.
	Clearly, it suffices to prove ${\rm Tr}(\beta L(\alpha))=0$ for any $\alpha\in F$. It can be seen
	\begin{align}\notag
		{\rm Tr}(\beta L(\alpha))&=\sum_{i=0}^{\ell-1}(\beta L(\alpha))^{q^i}\\
		&=\sum_{j=0}^t\sum_{i=0}^{\ell-1}(\beta\theta_j)^{q^i}\alpha^{q^{i\oplus j}}\notag\\
		&=\sum_{j=0}^t\Big(\sum_{k=0}^{\ell-1}(\beta\theta_j)^{q^{k\ominus j}}\alpha^{q^k}\Big)\label{eq266}\\
		&=\sum_{k=0}^{\ell-1}\Big(\sum_{j=0}^t(\beta\theta_j)^{q^{k\ominus j}}\Big)\alpha^{q^k},\notag
	\end{align}
	where $i\oplus j$ denotes $(i+j){\rm ~mod~}\ell$, and similarly, $k\ominus j$ denotes $(k-j){\rm ~mod~}\ell$. We use module $\ell$ here because $\alpha^{q^\ell}=\alpha$ for all $\alpha\in F$. Moreover, the equality (\ref{eq266}) follows by replacing $i\oplus j$ with $k$. It can be seen that for fixed $j$, $i\oplus j$ ranges over $[0,\ell-1]$ as $i$ runs in $[0,\ell-1]$. For simplicity, denote $\omega_k=\sum_{j=0}^t(\beta\theta_j)^{q^{k\ominus j}}$ for $k\in[0,\ell-1]$. It is easy to check that $\omega_{k\oplus 1}=\omega_k^q$. In particular, $\omega_t=\sum_{j=0}^t(\beta\theta_j)^{q^{t-j}}=0$ by the hypothesis. Therefore, $\omega_k=0$ for all $k\in[0,\ell-1]$, and thus ${\rm Tr}(\beta L(\alpha))=\sum_{k=0}^{\ell-1}\omega_k\alpha^{q^k}=0$. The claim is proved.
	
	By the system in Theorem 8 we know $\sum_{j=0}^t(\beta_i\theta_j)^{q^{t-j}}=0$ for $i\in[t]$. So by the claim above, it has ${\rm Im}(L)\subseteq\bigcap_{i=1}^t\beta_{i}^{-1}K$ and thus ${\dim}_B({\rm Im}(L))\leq {\rm dim}_B(\bigcap_{i=1}^t\beta_{i}^{-1}K)=\ell-t$, where the last equality is from Lemma \ref{lem12}. On the other hand, $L(x)$ is a $q$-polynomial over $F$ of degree $q^t$. It follows ${\dim}_B({\rm Im}(L))\geq \ell-t$. Therefore, it can only hold ${\dim}_B({\rm Im}(L))=\ell-t$. So ${\rm Im}(L)=\bigcap_{i=1}^t\beta_{i}^{-1}K$.
\end{proof}

\end{document}